\documentclass[11pt,a4paper]{amsart}

\usepackage{amsmath}
\usepackage{amsfonts}
\usepackage{amssymb}
\usepackage{amsthm}
\usepackage{epsfig}
\usepackage{graphicx}
\usepackage{url}
\usepackage[usenames,dvipsnames]{xcolor}
\usepackage[colorlinks=true,linkcolor=Blue,citecolor=Green,urlcolor=black]{hyperref}
\usepackage{nicefrac}
\usepackage{aliascnt}		
\usepackage[alwaysadjust]{paralist}
\usepackage[T1]{fontenc}
\usepackage[utf8]{inputenc}
\usepackage{tikz}
\usetikzlibrary{calc}
\usepackage{enumerate}

\newtheorem{thm}{Theorem}[section]
\newaliascnt{lem}{thm}
\newaliascnt{cor}{thm}
\newaliascnt{definition}{thm}
\newaliascnt{rem}{thm}
\newaliascnt{prop}{thm}
\newaliascnt{conj}{thm}
\newaliascnt{prob}{thm}

\newtheorem{lem}[lem]{Lemma}
\aliascntresetthe{lem}

\newtheorem*{lem*}{Lemma}

\newtheorem{cor}[cor]{Corollary}
\aliascntresetthe{cor}

\newtheorem*{cor*}{Corollary}

\newtheorem{definition}[definition]{Definition}
\aliascntresetthe{definition}

\newtheorem*{definition*}{Definition}

\newtheorem{rem}[rem]{Remark}
\aliascntresetthe{rem}

\newtheorem*{rem*}{Remark}

\newtheorem{prop}[prop]{Proposition}
\aliascntresetthe{prop}

\newtheorem*{prop*}{Proposition}

\newtheorem{conj}[conj]{Conjecture}
\aliascntresetthe{conj}

\newtheorem*{conj*}{Conjecture}

\newtheorem{prob}[prob]{Problem}
\aliascntresetthe{prob}

\newtheorem*{prob*}{Problem}

\newcommand{\N}{\mathbb{N}}
\newcommand{\Z}{\mathbb{Z}}

\newcommand{\R}{\mathbb{R}}

\newcommand{\one}{\mathbf{1}}

\newcommand{\cD}{\mathcal{D}}
\DeclareMathOperator{\vol}{vol}

\DeclareMathOperator{\conv}{conv}

\DeclareMathOperator{\poly}{poly}

\DeclareMathOperator{\inter}{int}

\DeclareMathOperator{\lin}{lin}

\DeclareMathOperator{\GL}{GL}

\newcommand{\norm}[1]{\|#1\|}

\newcommand{\df}{\mathrel{\mathop:}=}

\newcommand{\vc}{\mathcal{V}}

\newcommand{\cC}{\mathcal{C}}

\newcommand{\cF}{\mathcal{F}}

\newcommand{\fD}{\mathfrak{D}}

\numberwithin{equation}{section}

\begin{document}

\title{On compact representations of Voronoi cells of lattices}

\author[C. Hunkenschr\"oder]{Christoph Hunkenschr\"oder}
\address{Institute for Mathematics\\
         \'{E}cole Polytechnique F\'{e}d\'{e}rale de Lausanne\\
         CH-1015 Lausanne\\
         Switzerland}
\email{christoph.hunkenschroder@epfl.ch}

\author[G. Reuland]{Gina Reuland}
\address{Institute for Mathematics\\
         \'{E}cole Polytechnique F\'{e}d\'{e}rale de Lausanne\\
         CH-1015 Lausanne\\
         Switzerland}
\email{ginareuland@gmail.com}

\author[M. Schymura]{Matthias Schymura}
\address{Institute for Mathematics\\
         \'{E}cole Polytechnique F\'{e}d\'{e}rale de Lausanne\\
         CH-1015 Lausanne\\
         Switzerland}
\email{matthias.schymura@epfl.ch}

\thanks{This work was supported by the Swiss National Science Foundation (SNSF) within the project \emph{Convexity, geometry of numbers, and the complexity of integer programming (Nr. 163071)}. The paper grew out of the master thesis of the second author~\cite{reuland2018thesis}; an extended abstract of this work appeared as~\cite{hunkreulschym2019ipco}.}

\begin{abstract}
In a seminal work, Micciancio \& Voulgaris (2013) described a deterministic single-exponential time algorithm for the Closest Vector Problem (CVP) on lattices.
It is based on the computation of the Voronoi cell of the given lattice and thus may need exponential space as well.
We address the major open question whether there exists such an algorithm that requires only polynomial space.

To this end, we define a lattice basis to be $c$-compact if every facet normal of the Voronoi cell is a linear combination of the basis vectors using coefficients that are bounded by $c$ in absolute value.
Given such a basis, we get a polynomial space algorithm for CVP whose running time naturally depends on~$c$.
Thus, our main focus is the behavior of the smallest possible value of~$c$, with the following results:
There always exist $c$-compact bases, where $c$ is bounded by $n^2$ for an $n$-dimensional lattice; there are lattices not admitting a $c$-compact basis with $c$ growing sublinearly with the dimension; and every lattice with a zonotopal Voronoi cell has a $1$-compact basis.
\end{abstract}

\maketitle

\section{Introduction}

An $n$-dimensional lattice is the integral linear span of $n$ linearly independent vectors, $\Lambda = \{ Bz : z \in \Z^n \}$, $B \in \R^{d \times n}$.
If not stated otherwise, we always assume $d=n$, that is, the lattice has full rank.

Two widely investigated and important problems in the Algorithmic Geometry of Numbers,  Cryptography, and Integer Programming are the Shortest Vector Problem and the Closest Vector Problem.
Given a lattice~$\Lambda$, the Shortest Vector Problem (SVP) asks for a shortest non-zero vector in $\Lambda$. For a target vector $t \in \R^n$, the Closest Vector Problem (CVP) asks for a lattice vector~$z^\star$ minimizing the Euclidean distance $\| t-z \|$ from $t$ to a lattice point $z \in \Lambda$.

Let us recall the milestones of the algorithmic development regarding both SVP and CVP.
For a more detailed overview we refer to Hanrot, Pujol \& Stehl{\'e}~\cite{hanrotpujolstehle2011algorithms}, as well as to the more recent Gaussian Sampling approaches, for example, the one by Aggarwal \& Stephens-Davidowitz~\cite{aggarwal2018just}.

In the $1980$'s, Kannan presented algorithms solving SVP and CVP in running time $n^{\mathcal{O}(n)}$ and polynomial space~\cite{kannan1987minkowski}.
Although the constants involved in the running time had been improved, it took roughly fifteen years until a significantly better algorithm was discovered.
In $2001$, Ajtai, Kumar \& Sivakumar~\cite{ajtai2001sieve} gave a randomized algorithm for the Shortest Vector Problem, only taking $2^{\mathcal{O}(n)}$ time.
However, in addition to the randomness, they also had to accept exponential space dependency for their improved running time.
Though their algorithm is not applicable to the Closest Vector Problem in its full generality, they show in a follow-up work that for any fixed $\varepsilon$, it can be used to approximate CVP up to a factor of $(1+ \varepsilon)$ with running time depending on~$1/\varepsilon$~\cite{aks2}.
These authors posed the question whether randomness or exponential space is necessary for a running time better than $n^{\mathcal{O}(n)}$.

It took again around a decade until this question was partially answered by Micciancio \& Voulgaris~\cite{micciancio2013deterministic}, who obtained a deterministic $2^{\mathcal{O}(n)}$ algorithm for both problems.
Their algorithm is based on computing the Voronoi cell~$\vc_\Lambda$ of the lattice, the region of all points at least as close to the origin as to any other lattice point.
But as the Voronoi cell is a polytope with up to $2(2^n - 1)$ facets, the Micciancio-Voulgaris algorithm needs exponential space for storing the Voronoi cell in the worst (and generic) case.
Since storing the Voronoi cell in a different, ``more compact,'' way than by facet-description would lead to a decreased space requirement, they raise the question whether such a representation exists in general.

Our main objective is to propose such a compact representation of the Voronoi cell and to investigate its merits towards a single-exponential time and polynomial space algorithm for the CVP.
As being closer to the origin than to a certain lattice vector~$v$ expresses in the inequality $2 \, x^\intercal v \leq \| v \|^2$, the facets of $\vc_\Lambda$ can be stored as a set $\cF_\Lambda \subseteq \Lambda$ of lattice vectors, which are called the \emph{Voronoi relevant vectors}.
We say that a basis~$B$ of a lattice $\Lambda$ is $c$-\emph{compact}, if each Voronoi relevant vector of $\Lambda$ can be represented in $B$ with coefficients bounded by~$c$ in absolute value.
Hence, by iterating over $(2c + 1)^n$ vectors, we include the set $\cF_\Lambda$.
With $c(\Lambda)$, we denote the smallest $c$ such that there exists a $c$-compact basis of~$\Lambda$.
As a consequence of the ideas in~\cite{micciancio2013deterministic} and our notion of compactness we obtain (cf.~Corollary~\ref{cor:efficientWcompact}):
\begin{enumerate}[\indent $(i)$]
\item Given a $c$-compact basis of a lattice $\Lambda \subseteq \R^n$, we can solve the Closest Vector Problem in $(2c+1)^{\mathcal{O}(n)} \poly (n)$ time and polynomial space.
\end{enumerate}
Thus, the crucial question is: How small can we expect $c(\Lambda)$ to be for an arbitrary lattice?
If $c(\Lambda)$ is constant, then~(i) yields asymptotically the same running time as the initial Micciancio-Voulgaris algorithm, but uses only polynomial space.
Of course, this only holds under the assumption that we know a $c$-compact basis of~$\Lambda$.
This observation has consequences for the variant of CVP with preprocessing, which we discuss in Section~\ref{sec:algorithmic}.

As an example of a large family of lattices, we prove in Section~\ref{ssec:zonotopal-lattices}, that lattices whose Voronoi cell is a zonotope are as compact as possible:
%
\begin{enumerate}[\indent $(i)$]
\setcounter{enumi}{1}
\item If the Voronoi cell of $\Lambda$ is a zonotope, then $c(\Lambda) = 1$.
Moreover, a $1$-compact basis can be found among the Voronoi relevant vectors.
\end{enumerate}

Furthermore, every lattice of rank at most four has a $1$-compact basis (cf.~Corollary~\ref{cor:4-dim-compact}).
However, starting with dimension five there are examples of lattices with $c(\Lambda)>1$, and thus we want to understand how large this compactness constant can be in the worst case.
Motivated by applications in crystallography, the desire for good upper bounds on $c(\Lambda)$ was already implicitly formulated in~\cite{engel1988mathematical,engelmichelsenechal2001newgeometric}, and results of Seysen~\cite{seysen1999ameasure} imply that $c(\Lambda) \in n^{\mathcal{O}(\log n)}$.
We improve this to a polynomial bound and, on the negative side, we show that $c(\Lambda)$ may grow linearly with the dimension (Sections~\ref{ssec:upper-bound} \& \ref{ssec:lower-bound}):
\begin{enumerate}[\indent $(i)$]
\setcounter{enumi}{2}
 \item Every lattice possesses a basis that is $n^2$-compact.
 \item There exists a family of lattices $( \Lambda_n )_{n \geq 5}$ without a $o (n)$-compact basis.
\end{enumerate}

In Section~\ref{sec:relaxed}, we relax the notion of a $c$-compact basis as follows.
Denote by $\bar{c} (\Lambda)$ the smallest constant $\bar{c}$ such that there is \emph{any} square matrix $W$ with
\[
\cF_\Lambda \subseteq \{Wz : z \in \Z^n, \, \| z \|_\infty \leq \bar{c}\}.
\]
Hence, in general, the matrix $W$ generates a superlattice of $\Lambda$.
This relaxation is motivated by the fact that, given a basis, membership to a lattice can be checked in polynomial time.
Thus if $\bar{c} (\Lambda)$ is much smaller than $c(\Lambda)$, this additional check is faster than iterating over a larger set.
Our results regarding the relaxed compactness constant include the following:
\begin{enumerate}[\indent $(i)$]
\setcounter{enumi}{4}
 \item For every lattice $\Lambda$, we have $\bar{c}(\Lambda) \in \mathcal{O} (n \log n)$.
 \item There are lattices $\Lambda \subseteq \R^n$ with $c(\Lambda) \, / \, \bar{c}(\Lambda) \in \Omega(n)$.
\end{enumerate}

In summary, our contribution can be described as follows:
If we are given a $c(\Lambda)$-compact basis of a lattice, then we can modify the algorithm of Micciancio \& Voulgaris to obtain a polynomial space algorithm for CVP.
In whole generality, the time complexity of this algorithm cannot be better than $n^{\mathcal{O}(n)}$, as in Kannan's work.
However, we provide evidence that there are large and interesting classes of lattices, for which this improves to single-exponential time.
We think that it is worth to study the proposed compactness concept further.
In particular, it would be interesting to understand the size of the compactness constant for a generic lattice, and to conceive an efficient algorithm to find a $c$-compact basis.

\section{The notion of a \texorpdfstring{$c$}{c}-compact basis}

Given a lattice $\Lambda \subseteq \R^n$, its \emph{Voronoi cell} is defined by
\[
\vc_\Lambda = \left\{ x \in \R^n : \| x \| \leq \| x - z\| \text{ for all } z\in \Lambda \right\},
\]
where $\|\cdot\|$ denotes the Euclidean norm.
It consists of all points that are at least as close to the origin as to any other lattice point of~$\Lambda$.
The Voronoi cell turns out to be a centrally symmetric polytope having outer description $\vc_\Lambda = \left\{ x \in \R^n : 2 \, x^\intercal z \leq \|z\|^2 \text{ for all } z\in \Lambda \right\}$.
A vector $v \in \Lambda$ is called \emph{weakly Voronoi relevant} if the corresponding inequality $2 \, x^\intercal v \leq \|v\|^2$ defines a supporting hyperplane of $\vc_\Lambda$, and it is called \emph{strictly Voronoi relevant}, or simply \emph{Voronoi relevant}, if it is moreover facet-defining.
Let $\cF_\Lambda$ and $\cC_\Lambda$ be the set of strictly and weakly Voronoi relevant vectors of~$\Lambda$, respectively.
The central definition of this work is the following.

\begin{definition}
Let $\Lambda \subseteq \R^n$ be a lattice and let $c \in \N$.
A basis $B$ of $\Lambda$ is called $c$-\emph{compact}, if
\[
\cF_\Lambda \subseteq \left\{ Bz : z \in \Z^n, \ \| z \|_\infty \leq c \right\}.
\]
That is, each Voronoi relevant vector is a linear combination of the basis vectors with coefficients bounded by $c$ in absolute value.
Moreover, we define
\[
c(\Lambda) = \min\{c \geq 0 : \Lambda \text{ possesses a }c\text{-compact basis}\}
\]
as the \emph{compactness constant} of $\Lambda$.
\end{definition}

%
%
%

As discussed in the introduction, the notion of a $c$-compact basis provides a compact representation of the Voronoi cell $\vc_\Lambda$, the complexity of which depends on the value of the constant~$c$.
Before we set out to study the compactness constant in detail, we offer various equivalent definitions that serve as auxiliary tools and that also provide a better understanding of the underlying concept.

To this end, let $\Lambda^\star = \left\{y \in \R^n : y^\intercal z \in \Z \text{ for all } z \in \Lambda\right\}$ be the \emph{dual lattice} of $\Lambda$, and let $K^\star = \left\{x \in \R^n : x^\intercal y \leq 1 \text{ for all } y \in K\right\}$ be the \emph{polar body} of a compact convex set $K \subseteq \R^n$ containing the origin in its interior.
The basic properties we need are the following:
If $B$ is a basis of~$\Lambda$, then $B^{-\intercal} \df (B^{-1})^\intercal$ is a basis of $\Lambda^\star$, usually called the \emph{dual basis} of~$B$.
For a matrix $A \in \GL_n(\R)$ and a compact convex set $K$ as above, we have $(AK)^\star = A^{-\intercal} K^\star$.
To keep notation short, the convex hull of a finite point set $S \subseteq \R^n$ will be denoted by $\conv(S) = \{\sum_{s \in S} \alpha_s s \mid \forall s \in S: \alpha_s \in \R_{\geq 0}, \, \sum_{s \in S} \alpha_s = 1\}$, and the linear span of $S$ will be denoted by $\lin (S) = \{\sum_{s \in S} \alpha_s s \mid \forall s \in S: \alpha_s \in \R\}$.
We refer to Gruber's textbook~\cite{gruber2007convex} for details and more information on these concepts.

\begin{lem}
\label{lem:primal_dual}
Let $B = \{b_1, \ldots, b_n\}$ be a basis of a lattice $\Lambda \subseteq \R^n$.
The following are equivalent:
\begin{enumerate}[\indent $i)$]
 \item $B$ is $c$-compact,
 \item $c \cdot \conv (\cF_\Lambda)^\star$ contains the dual basis~$B^{-\intercal}$ of $\Lambda^\star$,
 \item writing $B^{-\intercal} = \{ b_1^\star, \ldots, b_n^\star \}$, we have
 \[
 \cF_\Lambda \subseteq \left\{ x \in \Lambda : |x^\intercal b_i^\star| \leq c \text{ for all }1 \leq i \leq n\right\},
 \]
 \item $\cF_\Lambda \subseteq c \, P_B$, where $P_B = \sum_{i=1}^n [-b_i, b_i]$.
\end{enumerate}
\end{lem}
\begin{proof}
$i) \Longleftrightarrow ii)$:
By definition, $B$ is $c$-compact if and only if $\cF_\Lambda \subseteq \{Bz : z \in \Z^n, \| z \|_\infty \leq c \}$.
This means that $Q = \conv (\cF_\Lambda) \subseteq B[-c,c]^n$.
Taking polars, we see that this is equivalent to $B^{-\intercal}\frac{1}{c}C_n^\star \subseteq Q^\star$, where $C_n^\star = \conv\{\pm e_1,\ldots,\pm e_n\}$ is the standard crosspolytope.
Since the columns of $B^{-\intercal}$ constitute a basis of the dual lattice $\Lambda^\star$, the proof is finished.

$i) \Longleftrightarrow iii)$:
$B = \{ b_1, \ldots, b_n \}$ is $c$-compact if and only if the representation $v = \sum_{i=1}^n \alpha_i b_i$ of any Voronoi relevant vector $v \in \cF_\Lambda$ satisfies $|\alpha_i| \leq c$, for all $1 \leq i \leq n$.
By the definition of the dual basis, we have $\alpha_i = v^\intercal b_i^\star$, which proves the claim.

$i) \Longleftrightarrow iv)$:
By definition, $\cF_\Lambda \subseteq c \, P_B$ if and only if for every $v \in \cF_\Lambda$, there are coefficients $\alpha_1, \ldots, \alpha_n \in \R$ such that $v = \sum_{i=1}^n \alpha_i b_i$ and $|\alpha_i| \leq c$.
These coefficients are unique, and since $B$ is a basis of $\Lambda$, they are integral, that is $\alpha_i \in \Z$.
Thus, the inclusion we started with is equivalent to saying that $B$ is $c$-compact.
\end{proof}

Part~iv) of the above lemma shows that the compactness constant $c(\Lambda)$ is the minimum $c$ such that $\cF_\Lambda \subseteq c \, P_B$, for some basis~$B$ of $\Lambda$.
In this definition, the concept has been introduced already by Engel, Michel \& Senechal~\cite{engelmichelsenechal2001newgeometric} together with the variant $\chi(\Lambda)$, where one replaces $\cF_\Lambda$ by the larger set $\cC_\Lambda$ of weakly Voronoi relevant vectors.
Motivated by applications in crystallography, a reoccurring question posed in~\cite{engel1988mathematical,engelmichelsenechal2001newgeometric} is to give good upper bounds on these lattice invariants $c(\Lambda)$ and $\chi(\Lambda)$.
Results of Seysen~\cite{seysen1999ameasure} on simultaneous lattice reduction of the primal and dual lattice imply that
\begin{align}
c(\Lambda) \leq \chi(\Lambda) \in n^{\mathcal{O}(\log n)}. \label{eqn:SeysenBound}
\end{align}
This is however the only bound that we are aware of.


\subsection{A polynomial upper bound}
\label{ssec:upper-bound}

In the sequel, we occassionally need Minkowski's \emph{successive minima} of a convex body $K$ and a lattice $\Lambda$ in $\R^n$.
For $1 \leq i \leq n$, the $i$th successive minimum is defined as
\[
\lambda_i(K,\Lambda) = \min \left\{ \lambda \geq 0 : \lambda K \text{ contains } i \text{ linearly independent points of } \Lambda \right\}.
\]
Minkowski's development of his Geometry of Numbers was centered around the study of these lattice parameters (we refer to Gruber's handbook~\cite{gruber2007convex} for background).
With this notion, Lemma~\ref{lem:primal_dual}~ii) provides a lower bound on the compactness constant of a given lattice.
Indeed, we have
\[
c(\Lambda) \geq \lambda_n(Q^\star,\Lambda^\star),
\]
where $Q = \conv(\cF_\Lambda)$.

Our first result aims for an explicit upper bound on $c(\Lambda)$ only depending on the dimension of the lattice.
To this end, we first need an auxiliary result.

\begin{lem}
\label{lem:MinkowskiBoundForOneDirection}
For a lattice $\Lambda \subseteq \R^n$ with Voronoi cell $\vc_\Lambda$ holds
$\lambda_1(\vc_\Lambda^\star ,\Lambda^\star) \leq \tfrac{2}{\pi} n$.
Hence, there exists a dual lattice vector $y^\star \in \Lambda^\star$ such that
\[
\vc_\Lambda \subseteq \left\{ x \in \R^n : |x^\intercal y^\star| \leq \tfrac{2}{\pi} n \right\}.
\]
\end{lem}
\begin{proof}
By Minkowski's fundamental theorem (cf.~\cite[Ch.~22]{gruber2007convex}), we have
\[
\lambda_1 (\vc_\Lambda, \Lambda) \leq 2 \left( \frac{\det(\Lambda)}{\vol(\vc_\Lambda)} \right)^{\frac{1}{n}}
	\quad \text{ and } \quad
\lambda_1 (\vc_\Lambda^\star, \Lambda^\star) \leq 2 \left( \frac{\det(\Lambda^\star)}{\vol(\vc_\Lambda^\star)} \right)^{\frac{1}{n}}.
\]
Moreover, by a result of Kuperberg~\cite[Cor.~1.6]{Kuperberg2008}, $\vol(K) \vol(K^\star) \geq \pi^n / n!$, for every centrally symmetric convex body $K \subseteq \R^n$.
Therefore,
\[
\lambda_1 (\vc_\Lambda, \Lambda) \lambda_1 (\vc_\Lambda^\star, \Lambda^\star) \leq 4 \left( \frac{\det(\Lambda)\det(\Lambda^\star)}{\vol(\vc_\Lambda) \vol(\vc_\Lambda^\star)} \right)^{\frac{1}{n}} \leq 4 \left( \frac{n!}{\pi^n} \right)^{\frac{1}{n}} \leq \frac{4}{\pi} n,
\]
since $\det(\Lambda)\det(\Lambda^\star) = 1$ (cf.~\cite[Ch.~1]{martinet2003perfect}).
The claimed bound now follows as $\lambda_i(\vc_\Lambda , \Lambda)=2$, for all $1 \leq i \leq n$.
\end{proof}

\begin{thm}
\label{thm:cn-upper-bound}
For every lattice $\Lambda \subseteq \R^n$, there exists an $n^2$-compact basis.
\end{thm}
\begin{proof}
We prove by induction on the dimension that there is a basis $D = \{y_1,\dots,y_n\}$ of $\Lambda^\star$ such that
\begin{align}
\label{eq:containment}
\vc_\Lambda &\subseteq \left\{x \in \R^n : \vert x^\intercal y_i \vert \leq \tfrac12 n^2, \ 1 \leq i \leq n \right\}.
\end{align}
Since every Voronoi relevant vector lies in the boundary of $2 \vc_\Lambda$, its inner product with each $y_i$ is then bounded by $n^2$.
Hence, the basis of $\Lambda$ that is dual to $D$ is an $n^2$-compact basis by Lemma~\ref{lem:primal_dual}~iii).

If $n=1$, the containment~\eqref{eq:containment} is trivially true, hence let $n \geq 2$.
Let $y_1$ be a shortest vector of $\Lambda^\star$ with respect to the norm $\| \cdot \|_{\vc_\Lambda^\star}$.
By Lemma~\ref{lem:MinkowskiBoundForOneDirection}, we have
$\vc_\Lambda \subseteq \left\{ x \in \R^n : |x^\intercal y_1| \leq \frac{2 n}{\pi} \right\}$.
Let $\Lambda^\prime = \Lambda \cap \{x \in \R^n : x^\intercal y_1 = 0\}$, and observe that the orthogonal projection $\pi: \R^n \rightarrow \{x \in \R^n : x^\intercal y_1 = 0\}$ fulfills $\pi(\Lambda^\star) = (\Lambda^\prime)^\star$, where we dualize with respect to the linear span of~$\Lambda^\prime$ (cf.~\cite[Ch.~1]{martinet2003perfect}).
By induction hypothesis, there is a basis $D^\prime = \{y^\prime_2,\dots, y^\prime_n\}$ of $(\Lambda^\prime)^\star$, such that
\[
\vc_{\Lambda^\prime} \subseteq \left\{x \in \R^n : x^\intercal y_1 = 0 \text{ and } \vert x^\intercal y^\prime_i \vert \leq \tfrac12 (n-1)^2, \ 2 \leq i \leq n \right\}.
\]
As $\Lambda^\prime \subseteq \Lambda$, we have $\vc_\Lambda \subseteq \vc_{\Lambda^\prime} + \lin\{y_1\}$.
Moreover, as $(\Lambda^\prime)^\star$ is the projection of $\Lambda^\star$ along $y_1$, there exist $\alpha_i \in [-1/2,1/2)$ such that $y_i = y^\prime_i + \alpha_i y_1 \in \Lambda^\star$ for $2 \leq i \leq n$, and $D = \{y_1, \dots, y_n\}$ is a basis of~$\Lambda^\star$.
Hence,
\begin{align*}
\vc_\Lambda &\subseteq \left\{ x \in \R^n : \vert x^\intercal y_1 \vert \leq \tfrac{2n}{\pi}, \ \vert x^\intercal y^\prime_i \vert \leq \tfrac12 (n-1)^2, \ 2 \leq i \leq n \right\} \\
&\subseteq \left\{ x \in \R^n : \vert x^\intercal y_1 \vert \leq \tfrac{2n}{\pi}, \ \vert x^\intercal y_i \vert \leq \tfrac12 (n-1)^2 + \tfrac{n}{\pi}, \ 2 \leq i \leq n \right\}. \\
&\subseteq \left\{ x \in \R^n : \vert x^\intercal y_i \vert \leq \tfrac12 n^2, \ 1 \leq i \leq n \right\},
\end{align*}
finishing the proof.
\end{proof}

\begin{rem}
Since also the weakly Voronoi relevant vectors $\cC_\Lambda$ lie in the boundary of $2\vc_\Lambda$, the basis from the previous proof also shows $\chi(\Lambda) \leq n^2$, for every lattice $\Lambda \subseteq \R^n$ (compare with~\eqref{eqn:SeysenBound}).
\end{rem}

Let us look at the constant $c(\Lambda)$ from a different angle.
A basis of a lattice is particularly nice if each Voronoi relevant vector is a $\{-1,0,1\}$-combination of the basis vectors.
As not every lattice possesses such a basis (see Proposition~\ref{prop:compactness} below), we relaxed the condition on the coefficients and introduced the lattice parameter $c(\Lambda)$, defined for all lattices.
Another way to relax the setting above is not to insist on a basis of $\Lambda$, but rather to look for a generating set $S$ such that each Voronoi relevant vector can be written as a $\{-1,0,1\}$-combination of the vectors in $S$.
In this setting, we are interested in finding a small set $S$. Such an $S$ of order $n \log n$ can be retrieved from an $n^2$-compact basis.

\begin{cor}
For every lattice $\Lambda \subseteq \R^n$ there exists a subset $S \subseteq \Lambda$ of cardinality $\mathcal{O} (n \log n)$ such that
\[
\cF_\Lambda \subseteq \bigg\{ \sum_{s \in S} \sigma_s \, s : \sigma_s \in \{-1,0,1\}, \text{ for } s \in S \bigg\}.
\]
\end{cor}
\begin{proof}
By Theorem~\ref{thm:cn-upper-bound}, there exists a $c$-compact basis $B$ of $\Lambda$ with $c \leq n^2$. Let $M \df \lfloor \log_2 c \rfloor$.
Each $0 \leq \alpha \leq c$ can be written as $\alpha = \sum_{j=0}^M 2^j \sigma_j$, for some unique $\sigma_j \in \{0,1\}$.
For each vector $b_i \in B$ and $0 \leq j \leq M$, we define the vector
\[
s_{i,j} \df 2^j b_i.
\]
This gives $\mathcal{O}(n \log_2 (n^2)) = \mathcal{O} (n \log n)$ vectors in total, and clearly every vector $v = \sum_{i=1}^n \alpha_i b_i$ with $\vert \alpha_i \vert \leq c$ can be written as a linear combination of the $s_{i,j}$ using only coefficients in $\{-1,0,1\}$.
\end{proof}

\begin{rem}
With a different method Daniel Dadush (personal communication) proves that the subset~$S$ can be chosen to consist of \emph{Voronoi relevant vectors} itself.
\end{rem}

\subsection{Lattices without sublinearly-compact bases}
\label{ssec:lower-bound}

In this part, we identify an explicit family of lattices whose compactness constant grows at least linearly with the dimension. 
This requires some technical work; the pure existence of such a family also follows from Proposition~\ref{prop:lambda<2}~iii) below.
However, based on the understanding of the lattice discussed in this section, we are able to discriminate between the compactness constant and a relaxed variant, which will be introduced in the next section.

For any $a \in \N$ and $n \in \N$, we define the lattice
\begin{align}
\Lambda_n(a) &= \left\{ z \in \Z^n : \ z_1 \equiv \dots \equiv z_n \!\mod a \right\}, \label{eqn:Lambda_n}
\end{align}
whose dual lattice is given by
\begin{align}
\Lambda_n(a)^\star &= \left\{ z \in \tfrac{1}{a} \Z^n : \one^\intercal z \in \Z \right\}, \label{eqn:dualLambda_n}
\end{align}
where $\one=(1,\ldots,1)^\intercal$ denotes the all-one vector.
The special structure of these lattices allows us to write down the Voronoi relevant vectors explicitly.

\begin{lem}
\label{lem:vor-rel-vectors}
Let $n \in \N_{\geq 4}$, $a = \lceil n/2 \rceil$, and write $\Lambda_n = \Lambda_n(a)$.
Then, a vector $v \in \Lambda_n$ is strictly Voronoi relevant if and only if either $v = \pm \mathbf{1}$, or there exists an index set $\emptyset \neq S \subsetneq \{1,\dots,n\}$ such that
\begin{align}
\label{eq:shape_vectors}
v_i &= \begin{cases} a-\ell & i \in S \\ -\ell & i \notin S \end{cases}, \qquad \text{and } \quad \ell \in \left\{ \left\lfloor \frac{a\vert S \vert}{n} \right\rfloor, \left\lceil \frac{a\vert S \vert}{n} \right\rceil \right\} .
\end{align}
\end{lem}
\begin{proof}
Let us first discuss the vectors $\pm \mathbf{1}$. They have squared norm $n$, and if there is a shorter vector $v$, it must contain zero coordinates. But due to the definition of $\Lambda_n$, all its coordinates are then multiples of $a$, so it has squared norm at least $a^2 \geq n^2/4 \geq n$ for $n \geq 4$. Hence, $\pm \mathbf{1}$ are shortest vectors of the lattice and therefore always strictly Voronoi relevant.
As we are only interested in the strictly Voronoi relevant vectors in this proof, we will omit the word ``strictly'' henceforth.

Voronoi characterized a Voronoi relevant vector~$v$ in a lattice~$\Lambda$ by the property that $\pm v$ are the only shortest vectors in the co-set $v + 2 \Lambda$ (cf.~\cite[p.~477]{conwaysloane1999splag}).
We use this crucially to show that Voronoi relevant vectors different from $\pm \one$ are characterized by~\eqref{eq:shape_vectors}.

\underline{$v$ Voronoi relevant $\Rightarrow$ $v$ of Shape~\eqref{eq:shape_vectors}:}
Let $v \neq \pm \one$ be Voronoi relevant.
We have $v \in [-a,a]^n$, as $2a \, e_i \in 2 \Lambda_n$ otherwise implies that $v$ is not a shortest vector in its co-set $v + 2 \Lambda_n$.
Let us first assume that there is an index $i$ such that $v_i \in \{0,\pm a\}$.
By definition of $\Lambda_n$, we have $v_i \equiv v_j \!\mod a$, for all $j$, hence $v \in \{0,\pm a\}^n$.
If $v$ has at least two non-zero coordinates, let $\widetilde{v}$ arise from $v$ by changing the sign of exactly one of them.
Observe that $\widetilde{v}$ is linearly independent from~$v$, has the same length, and is contained in $v + 2\Lambda_n$.
This contradicts the assumption that~$v$ was Voronoi relevant.
If $v$ has only one non-zero entry, say $v_j \neq 0$, then it is of Shape~\eqref{eq:shape_vectors}.
Indeed, we can either take $S=\{j\}$ and $\ell=0$, or $S=\{1,\ldots,n\}\setminus\{j\}$ and $\ell = \lceil a(n-1)/n \rceil = a$.

This leaves us with the case $v \in [-(a-1),a-1]^n$.
Again, by the definition of $\Lambda_n$, there is an integer $1 \leq r \leq a-1$ such that $v \in \{a-r,-r\}^n$.
Let~$k$ be the number of entries of $v$ that are equal to~$a-r$.
Note that $1 \leq k \leq n-1$ as otherwise $v= \pm \mathbf{1}$.
For the norm of $v$, we obtain
\[
\norm{v}^2 = n r^2 - 2akr + ka^2.
\]
Seen as a rational quadratic function in $r$, it is minimized for $r^\prime = ak/n$.
As increasing or decreasing $r$ by $2$ corresponds to adding or subtracting $2 \cdot \one \in 2\Lambda_n$ to $v$, we must have $r \in [r^\prime -1, r^\prime + 1]$.
If $r^\prime$ is not integral, this corresponds to $r \in \{\lceil ak/n \rceil, \lfloor ak/n \rfloor\}$.
If $r^\prime$ is integral, observe that $r = r^\prime \pm 1$ corresponds to two linearly independent vectors in the same co-set and of the same length,
hence again $r = r^\prime \in \{\lceil ak/n \rceil, \lfloor ak/n \rfloor\}$,
so that~$v$ is indeed of Shape~\eqref{eq:shape_vectors}.

\underline{$v$ of Shape~\eqref{eq:shape_vectors} $\Rightarrow$ $v$ Voronoi relevant:}
For the other direction,
let $v$ be a vector of Shape~\eqref{eq:shape_vectors} with index set $S$ and parameter $\ell$.
Let $u \in v + 2\Lambda_n$ be a shortest vector within the co-set $v + 2\Lambda_n$.
We claim that $u = \pm v$, which will prove that~$v$ is Voronoi relevant.
To this end, recall from above that necessarily $u \in [-a,a]^n$.
Moreover, as $u-v \in 2\Lambda_n$, we have $v_i - v_j \equiv u_i - u_j \!\mod 2a$.
Therefore, if there are indices $i\neq j$ such that $v_i = v_j$, then we have $u_i \equiv u_j \!\mod 2a$. Unless we are in the extreme case $u \in \{0,\pm a\}^n$ (see Case (a)), this even implies $u_i = u_j$ (see Case (b)).

Case (a): We make a second case distinction depending on the number of non-zero entries of $u$. This number is always either equal to $\vert S \vert$ or $n-|S|$.

Note that the case of $u$ having exactly $1$ non-zero entry (i.e.\ $|S| \in \{1,n-1\}$) will be covered by Case (b) below.

If $u$ has at least $3$ non-zero entries ($|S| \in \{3,4,\dots,n-3\}$), observe that the vector $u^\prime=(u_1^\prime,\ldots,u_n^\prime)^\intercal$ defined by $u_i^\prime = \vert u_i \vert - 2$ is in the same co-set, but also shorter than $u$, a contradiction.

For the last case, $u$ having two non-zero entries, the vector $u^\prime$ as above is only strictly shorter if $n$ is odd. If $n$ is even however, $u$ and $u^\prime$ will have the same norm. In this particular case, observe that $a|S|/n \in \{1,a-1\}$, hence $\ell = a|S|/n$, as we do not round.
But this is a contradiction, as $u$ and $v$ differ by $\mathbf{1} \notin 2 \Lambda_n$, that is, they are not in the same co-set.

Case (b): Henceforth, whenever $v_i = v_j$, we have $u_i = u_j$.
By possibly switching to $u^\prime = -u$, we can assume that for some $0 \leq r \leq a$, $u_i = a-r$ for $i \in S$ and $u_j = -r$ for $j \notin S$. This is, $u_i = a-r$ whenever $v_i = a - \ell$ and $u_j = -r$ whenever $v_j = - \ell$.
For the norm of $u$, we obtain
\[
\|u\|^2 = |S|(a-r)^2 + (n-|S|) r^2 = n r^2 - 2ar|S| + |S|a^2.
\]
Seen as a rational quadratic function in $r$, this term is uniquely minimized for $\hat{r} = a|S|/n$. Observe that there may be two choices for $\ell$, $\lfloor \hat{r}\rfloor, \lceil \hat{r} \rceil$.
It is clear that $r$ also has to be one of these values, as otherwise $u$ is not a shortest vector in its co-set.
But observe that the two choices lead to two vectors whose difference is $\mathbf{1} \notin 2\Lambda_n$. As $u$ and $v$ have to be in the same co-set, we have $u = \pm v$, since we may have switched to $-u$ in the beginning.
\end{proof}

\begin{thm}
\label{thm:lower_bound}
Let $n \in \N_{\geq 4}$, $a = \lceil n/2 \rceil$.
Then, the lattice $\Lambda_n = \Lambda_n(a)$ has compactness constant $c(\Lambda_n) \geq \left\lceil \frac{n}{4} \right\rceil$.
\end{thm}
\begin{proof}
For brevity, we write $c = c(\Lambda_n)$, $Q = \conv(\cF_{\Lambda_n})$.
As $\one \in \Lambda_n$, there exists a $w \in \Lambda_n^\star$ with $\one^\intercal w = 1$, for instance, take $w=e_1$.
This implies that each basis of $\Lambda_n^\star$ contains a vector $y$ such that $\one^\intercal y$ is an odd integer.
In particular, using the characterization of Lemma~\ref{lem:primal_dual}, we know that $c \, Q^\star$ has to contain such a~$y$.
As~$Q^\star$ is centrally symmetric, assume $\one^\intercal y \geq 1$.
Further, since $\Lambda_n^\star$ is invariant under permutation of the coordinates, assume the entries of $y$ are ordered non-increasingly,
\begin{align}
y_1 \geq y_2 \geq \dots \geq y_n.\label{eqn:orderony}
\end{align}
Let us outline our arguments first:
We split $\one^\intercal y$ into two parts, by setting $A \df \sum_{i=1}^{k} y_i$, and $B \df \sum_{i>k}^n y_i$, where $k = \lceil n/2 \rceil$.
We show that $A \geq B + 1$, and construct a Voronoi relevant vector $v \in \Lambda_n$ whose first $k$ entries are roughly $n/4$, and its last $n-k$ entries are roughly $-n/4$ by using Lemma~\ref{lem:vor-rel-vectors} and choosing $S = \{1,\dots,k\}$, $\ell = \lfloor ak/n \rfloor = \lfloor a^2/n \rfloor$.
We then obtain $v^\intercal y \approx \frac{n}{4} A - \frac{n}{4} B \geq n/4$ by carefully distinguishing the four cases $n \!\mod 4$.

For showing $A \geq B + 1$, consider $y_k$.
As $y \in \Lambda_n^\star$, there is an integer~$z$ such that we can write $y_k = \frac{z}{a}$.
We can assume $z > 0$, since otherwise $B \leq 0$ and we are done since $A+B = \one^\intercal y \geq 1$.
Note that we have $A \geq k y_k = z$ and $B \leq (n-k) \frac{z}{a} \leq z$ by~\eqref{eqn:orderony}.
Let $\alpha, \gamma \geq 0$ such that $A = z + \alpha$ and $B= z - \gamma$. As $A+ B = 2z + \alpha - \gamma$ has to be an odd integer, we have $\vert \alpha - \gamma \vert \geq 1$, implying $\alpha \geq 1$ or $\gamma \geq 1$. 
Therefore, in fact we have $A \geq \max \{ B  + 1 , 1 \}$.

Using this inequality and carefully evaluating $v^\intercal y = (a - \ell)A - \ell B$ for the four cases $n \!\mod 4$, the claim follows.

Recall that we construct the Voronoi relevant vector $v$ by choosing $k = a = \lceil n/2 \rceil$, $S=\{1,\dots,k\}$, $\ell = \lfloor ak/n \rfloor = \lfloor a^2/n \rfloor$, and applying Lemma~\ref{lem:vor-rel-vectors}.

We obtain $v^\intercal y = (a - \ell)A - \ell B$, and are ready to distinguish the four cases $n \!\mod 4$.
\begin{enumerate}
\item $n = 4m$. Hence, we have $a = k = 2m$, and $\ell = m$. Thus,
\[
v^\intercal y = (a - \ell)A - \ell B = m(A-B) \geq m =n/4.
\]
\item $n = 4m + 1$. Hence, we have $a = k = 2m + 1$, and $\ell = m$. Thus,
\[
v^\intercal y = (a - \ell)A - \ell B = m(A - B) + A \geq m + 1 \geq n/4.
\]
\item $n = 4m + 2$. Hence, we have $a = k = 2m + 1$, and $\ell = m$. Thus,
\[
v^\intercal y = (a - \ell)A - \ell B = m(A - B) + A \geq m + 1 \geq n/4.
\]
\item $n = 4m + 3$. Hence, we have $a = k = 2m + 2$, and $\ell = m+1$. Thus,
\[
v^\intercal y = (a - \ell)A - \ell B = (m + 1)(A - B) \geq m+1 \geq n/4.
\]
\end{enumerate}
As the constant $c$ is integral, the claim follows.
\end{proof}

\subsection{Compact bases and zonotopal lattices}
\label{ssec:zonotopal-lattices}

For the sake of brevity, we call a $1$-compact basis of a lattice just a \emph{compact basis}.
A class of lattices that allow for a compact representation of their Voronoi cells are the lattices of \emph{Voronoi's first kind}.
They correspond to those lattices~$\Lambda$ that constitute the first reduction domain in Voronoi's reduction theory (see~\cite{vallentin2003thesis,voronoi1908deuxieme}).
These lattices have been characterized in~\cite{conwaysloane1992lowdim} by possessing an \emph{obtuse superbasis}, which is a set of vectors $\{b_0,\ldots,b_n\} \subseteq \Lambda$ that generates $\Lambda$, and that fulfills the superbasis condition $b_0 + \ldots + b_n = 0$ and the obtuseness condition $b_i^\intercal b_j \leq 0$, for all $i \neq j$.
Given an obtuse superbasis, for each Voronoi relevant vector $v \in \Lambda$ there is a strict non-empty subset $S \subseteq \{0,1,\dots,n\}$ such that $v = \sum_{i \in S} b_i$.

Let us compare lattices of Voronoi's first kind with lattices possessing a compact basis.
\pagebreak
\begin{prop}
\label{prop:compactness}\ 
\begin{enumerate}[\indent $i)$]
 \item Every lattice of Voronoi's first kind has a compact basis.

 \item Every lattice of rank at most three has a compact basis.
 
 \item For $n \geq 4$, the \emph{checkerboard lattice} $D_n = \{x \in \Z^n : \mathbf{1}^\intercal x \in 2\Z\}$ is not of Voronoi's first kind, but has a compact basis.
 
 \item There exists a lattice $\Lambda \subseteq \R^5$ with $c(\Lambda) \geq 2$.
\end{enumerate}
\end{prop}
\begin{proof}
$i)$: Every obtuse superbasis contains in fact a compact basis.
Indeed, using the representation of a Voronoi relevant vector above and writing $b_0 = - \sum_{i=1}^n b_i$, we get $v=\sum_{i \in S} b_i = - \sum_{i \notin S} b_i$.
One of the terms does not use~$b_0$.

$ii)$: Every lattice of dimension at most three is of Voronoi's first kind (cf.~\cite{conwaysloane1992lowdim}), so part~i) applies.

$iii)$: Bost \& K{\"u}nnemann~\cite[Prop.~B.2.6]{bost2010hermitian} showed that for $n \geq 4$, the lattice~$D_n$ is not of Voronoi's first kind.
One can easily verify that the set $B=\{b_1,\dots,b_n\}$ with
$b_1 = e_1 + e_n$, and $b_i = e_i - e_{i-1}$ for $2 \leq i \leq n$, is a basis of~$D_n$.
Observing that the vectors $2e_i \pm 2e_j$ are in $2 D_n$ for all $i,j$, a vector~$v$ that is the unique (up to sign) shortest vector in the co-set $v + 2\Lambda$, must be of the form $\{\pm (e_i \pm e_j) : 1 \leq i < j \leq n \}$.
A routine calculation shows that all these vectors are a $\{-1,0,1\}$-combination of the basis $B$.

$iv)$: This follows from Theorem~\ref{thm:lower_bound} with the lattice $\Lambda_5(3)$.
\end{proof}

We now explore to which extent these initial observations on lattices with compact bases can be generalized.

A \emph{zonotope} $Z$ in $\R^n$ is a Minkowski sum of finitely many line segments, that is, $Z = \sum_{i=1}^r [a_i,b_i]$, for some $a_i,b_i \in \R^n$.
The vectors $b_1-a_1,\ldots,b_r-a_r$ are usually called the \emph{generators} of~$Z$.
We call a lattice \emph{zonotopal} if its Voronoi cell is a zonotope.
A generic zonotopal lattice has typically high combinatorial complexity.
An explicit example is the root lattice~$A_n^\star$; its zonotopal Voronoi cell is generated by $\binom{n+1}{2}$ vectors and it has exactly the maximum possible $2(2^n - 1)$ facets (cf.~\cite[Ch.~4 \& Ch.~21]{conwaysloane1999splag}).
However, not every generic lattice is zonotopal.
For instance, a perturbation of the~$E_8$ root lattice gives a generic non-zonotopal lattice (cf.~\cite[Sect.~4]{erdahlryshkov1994latticedicing}).

It turns out that every lattice of Voronoi's first kind is zonotopal, but starting from dimension four, the class of zonotopal lattices is much richer (cf.~Vallentin's thesis~\cite[Ch.~2]{vallentin2003thesis} and~\cite{erdahlryshkov1994latticedicing}).
In the following, we prove that every zonotopal lattice possesses a compact basis, thus extending Proposition~\ref{prop:compactness}~$i)$ significantly.

Our proof relies on the beautiful work of Erdahl~\cite{erdahl1999zonotopes} who unraveled an intimate relationship between zonotopal lattices and so-called dicings.
A \emph{dicing}~$\fD$ in $\R^n$ is an arrangement of hyperplanes consisting of at least~$n$ families of infinitely many equally-spaced hyperplanes with the following properties:
\begin{enumerate}[ (i)]
 \item There are~$n$ families with linearly independent normal vectors.
 \item Every vertex of $\fD$ is contained in a hyperplane of each family.
\end{enumerate}
The interesting cases are those with more than~$n$ families of hyperplanes.

It turns out that the vertex set of a dicing forms a lattice, denoted by $\Lambda(\fD)$.
Indeed, the vertex set induced by the~$n$ linearly independent families forms a lattice, and because of property~(ii) no additional vertices are introduced by the remaining families.
A basis of the lattice $\Lambda(\fD)$ may be obtained from taking the inverse of the matrix whose rows are $n$ linearly independent normal vectors appropriately scaled (they exist by property~(i)).

Erdahl~\cite[Thm.~3.1]{erdahl1999zonotopes} shows that a dicing $\fD$ can be represented by a set $D = \{\pm d_1,\ldots,\pm d_r\}$ of hyperplane normals and a set $E = \{\pm e_1,\ldots,\pm e_s\} \subseteq \Lambda(\fD)$ of edge vectors of the arrangement~$\fD$ satisfying:

\begin{enumerate}[\indent E1)]
 \item Each pair of edges $\pm e_j \in E$ is contained in a line $d_{i_1}^\perp \cap \ldots \cap d_{i_{n-1}}^\perp$, for some linearly independent $d_{i_1},\ldots,d_{i_{n-1}} \in D$, and conversely each such line contains a pair of edges.

 \item For each $1 \leq i \leq r$ and $1 \leq j \leq s$, we have $d_i^\intercal e_j\in \{0,\pm 1\}$.
\end{enumerate}

\noindent For clarity we denote the dicing by $\fD=\fD(D,E)$.

\begin{thm}
\label{thm:zonotopal}
Every zonotopal lattice has a compact basis.
It can be found among its Voronoi relevant vectors.
\end{thm}
\begin{proof}
We start by reviewing the \emph{Delaunay tiling} of a lattice~$\Lambda$.
A sphere $B_c(R) = \{x \in \R^n : \|x - c\|^2 \leq R^2\}$ is called an \emph{empty sphere} of~$\Lambda$ (with center $c \in \R^n$ and radius $R \geq 0$), if every point in $B_c(R) \cap \Lambda$ lies on the boundary of $B_c(R)$.
A \emph{Delaunay polytope} of $\Lambda$ is defined as the convex hull of $B_c(R) \cap \Lambda$, where $B_c(R)$ is an empty sphere.
The family of all Delaunay polytopes induces a tiling $\cD_\Lambda$ of $\R^n$ which is the Delaunay tiling of~$\Lambda$.
This tiling is in fact dual to the Voronoi tiling.

Erdahl~\cite[Thm.~2]{erdahl1999zonotopes} shows that the Voronoi cell of a lattice is a zonotope if and only if its Delaunay tiling is a dicing.
More precisely, the tiling $\cD_\Lambda$ induced by the Delaunay polytopes of~$\Lambda$ is equal to the tiling induced by the hyperplane arrangement of a dicing $\fD = \fD(D,E)$ with normals $D = \{\pm d_1,\ldots,\pm d_r\}$ and edge vectors $E = \{\pm e_1,\ldots,\pm e_s\}$.
By the duality of the Delaunay and the Voronoi tiling, an edge of $\cD_\Lambda$ containing the origin corresponds to a facet normal of the Voronoi cell $\vc_\Lambda$.
Therefore, the edge vectors $E$ are precisely the Voronoi relevant vectors of~$\Lambda$.

Now, choosing $n$ linearly independent normal vectors, say $d_1,\ldots,d_n \in D$, the properties E1) and E2) imply the existence of edge vectors, say $e_1,\ldots,e_n \in E$, such that $d_i^\intercal e_j = \delta_{ij}$, with $\delta_{ij}$ being the Kronecker delta.
Moreover, the set $B = \{e_1,\ldots,e_n\}$ is a basis of $\{x \in \R^n : d_i^\intercal x \in \Z, 1 \leq i \leq n \}$, which by property~E2) equals the whole lattice~$\Lambda$.
Hence, $\{d_1,\ldots,d_n\}$ is the dual basis of~$B$ and every Voronoi relevant vector $v \in \cF_\Lambda = E$ fulfills $d_i^\intercal v \in \{0,\pm 1\}$.
In view of Lemma~\ref{lem:primal_dual}~$iii)$, this means that~$B$ is a compact basis of~$\Lambda$ consisting of Voronoi relevant vectors, as desired.
\end{proof}

\subsection{Compact bases in small dimensions}
\label{ssec:dimensions-four}

We have seen in Proposition~\ref{prop:compactness} that every lattice of rank at most three has a compact basis, and that there are five-dimensional lattices without compact bases.
In the sequel we complete the picture and show that every four-dimensional lattice admits a compact basis as well.

Our argument uses tools from the theory of parallelotopes which requires to set up the compactness constant in this more general framework.
For details and background on the following definitions and statements on parallelotopes we refer to~\cite[\S 32]{gruber2007convex}.
A \emph{parallelotope} (also called parallelohedron) is a convex polytope~$P \subseteq \R^n$ that admits a facet-to-facet tiling of~$\R^n$ by translations.
Voronoi cells of lattices are prime examples of parallelotopes.
Every parallelotope is centrally symmetric, and we may assume that its center of symmetry is at the origin.
The set of translation vectors that constitute the facet-to-facet tiling by copies of~$P$ is in fact a lattice, and we denote it by~$\Lambda(P)$.
Every facet~$F$ of~$P$ corresponds to a lattice vector $x \in \Lambda(P)$ such that $P \cap (P+x) = F$.
Such a lattice vector is called a \emph{facet vector}.
More generally, a lattice vector $x \in \Lambda(P)$ such that $P \cap (P+x)$ is a face of both~$P$ and $P+x$ is called a \emph{standard vector} of~$P$.

For Voronoi cells the facet vectors and the standard vectors are exactly the strictly and weakly Voronoi relevant vectors, respectively.
Therefore, we can extend our notation from the previous sections from Voronoi cells and lattices, to general parallelotopes:
We write $\cF_P$ and $\cC_P$ for the set of facet vectors and standard vectors of~$P$, respectively, and $c(P)$ and $\chi(P)$ for the corresponding compactness constants.
For example, $\chi(P)$ is the minimal $\chi > 0$ such that there is a basis $B=\{b_1,\ldots,b_n\}$ of $\Lambda(P)$ with the property that every standard vector $x \in \cC_P$ can be written as $x = \sum_{i=1}^n \gamma_i b_i$, for some $|\gamma_i| \leq \chi$.

With this notation we prove the crucial fact, that if a parallelotope~$Q$ decomposes into the Minkowski sum of another parallelotope $P$ and a (possibly lower-dimensional) zonotope~$Z$, then $\chi(Q) \leq \chi(P)$.
We write $Z(U) = \sum_{i=1}^r [-u_i,u_i]$ for the zonotope spanned by the set of vectors $U = \{u_1,\ldots,u_r\}$.

\begin{prop}
\label{prop:zonotopal-part}
Let $Q \subseteq \R^n$ be a parallelotope that admits a decomposition $Q = P + Z(U)$, for some parallelotope~$P$, and a finite set of vectors $U \subseteq \R^n$.
Then, there is a linear map $\varphi: \R^n \rightarrow \R^n$ with $\varphi(\Lambda(P)) = \Lambda(Q)$ satisfying
\goodbreak
\begin{enumerate}[\indent $i)$]
\item For $x \in \cC_Q$, we have $\varphi^{-1}(x) \in \cC_P$.
\item For $x \in \cF_P$, we have $\varphi (x) \in \cF_Q$.
\end{enumerate}
In particular, $\chi(Q) \leq \chi(P)$.
\end{prop}

\begin{proof}
First note that if $P + Z(U)$ is a parallelotope, then every vector $u \in U$ is a \emph{free} vector for~$P$, that is, $P + [-u,u]$ is a parallelotope as well (cf.~\cite{dutourgrishukhinmagazinov2014onthesum}).
We thus get a chain of parallelotopes $P = P_0 \subseteq P_1 \subseteq \ldots \subseteq P_r=Q$, where $P_i = P_{i-1} + [-u_i,u_i]$, for $1 \leq i \leq r$, and $U = \{u_1,\ldots,u_r\}$.
By induction on $r$ it thus suffices to consider the case $r=1$.

Hence, let $Q = P + [-u,u]$, for some non-zero vector $u \in \R^n$.
Dutour Sikiri\'{c} et al.~\cite[Lem.~1 \& Lem.~3]{dutourgrishukhinmagazinov2014onthesum} give a characterization of the standard vectors of~$Q$ in terms of those of~$P$:
First, there is a dual lattice vector $e_u \in \Lambda(P)^\star$ such that $\Lambda(Q) = A_u \Lambda(P)$, where $A_u x = x + 2 (e_u^\intercal x) u$, for $x \in \R^n$.
Then, $z = A_u w \in \Lambda(Q)$ is a standard vector of~$Q$ if and only if~$w$ is a standard vector of~$P$, and $e_u^\intercal w \in \{0,\pm 1\}$.

This implies that $\varphi(x)=A_u x$ is a bijection between the lattices $\Lambda(P)$ and $\Lambda(Q)$ satisfying~$i)$.
Moreover, the proof of~\cite[Lem.~1]{dutourgrishukhinmagazinov2014onthesum}
reveals that $A_u$, and thus~$\varphi$, sends facet vectors to facet vectors,
hence~$ii)$ holds as well.
For $r\geq2$, we define $\varphi$ inductively by setting $\varphi(x) = A_{u_r} \cdot\ldots\cdot A_{u_1} x$.

Finally we show that $\chi(Q) \leq \chi(P)$.
As just observed, any basis~$B$ of~$\Lambda(P)$ is sent to a basis $\varphi(B)$ of $\Lambda(Q)$.
Moreover, a standard vector $y = \sum_{i=1}^n \alpha_i \varphi(b_i) \in \cC_Q$ represented in the basis $\varphi(B)$ corresponds to a standard vector $\varphi^{-1}(y) = \sum_{i=1}^n \alpha_i b_i \in \cC_P$ using the same coefficients when represented in the basis $B$.
Thus, if every vector in $\cC_P$ can be represented in $B$ with coefficients bounded by $\chi(P)$, the same holds for all vectors in $\cC_Q$ with respect to $\varphi(B)$.
\end{proof}

As a consequence, we get that zonotopal parallelotopes $Z$ allow for a compact representation even of the set $\cC_Z$
which strengthens Theorem~\ref{thm:zonotopal}.
In particular, every three-dimensional parallelotope has this property (cf.~\cite[\S 32.2]{gruber2007convex}).

\begin{cor}
\label{cor:zonotopes-chi-compact}
Let $Z$ be a parallelotope that is a zonotope.
Then $\chi(Z)=1$.
\end{cor}

\begin{proof}
There is a set of vectors $U'=\{u_1,\ldots,u_m\}$ such that $Z=Z(U')$.
We may assume that $u_1,\ldots,u_n$ are linearly independent, and we write $P=[-u_1,u_1]+\ldots+[-u_n,u_n]$.
Then, $Z = P + Z(U)$, for $U=U' \setminus\{u_1,\ldots,u_n\}$, and since $P$ is a parallelepiped, it is actually a parallelotope.

Thus, Proposition~\ref{prop:zonotopal-part} implies that $\chi(Z) \leq \chi(P)$ and it suffices to show that $\chi(P)=1$ for every parallelepiped~$P$.
The standard vectors of~$P$ are exactly those $x\in\Lambda(P) \setminus \{0\}$ such that $P \cap (x+P) \neq \emptyset$.
Writing $\pm f_1,\ldots,\pm f_n$ for the~$n$ pairs of facet vectors of~$P$, we find that $\{f_1,\ldots,f_n\}$ is a basis of $\Lambda(P)$ in which every standard vector admits a $\{0,\pm 1\}$-representation.
\end{proof}

We now focus again on parallelotopes that are Voronoi cells but work in the more convenient language of quadratic forms.
A famous conjecture of Voronoi states that every parallelotope is an affine image of a Voronoi cell of a lattice (cf.~\cite[\S 32]{gruber2007convex}).
As long as this is not settled we need to make the distinction.

Let $q : \R^n \to \R$ be a positive definite quadratic form defined by $q(x) = x^\intercal A^\intercal A x$, for some invertible matrix $A \in \R^{n \times n}$.
We associate the lattice $\Lambda=A\Z^n$ to~$q$.
Analogously to the lattice case, the Voronoi cell of~$q$ is defined as
\[
\vc_q = \left\{ x \in \R^n : q(x) \leq q(x-z)\text{ for all }z \in \Z^n \right\}.
\]
This is a linear image of the Voronoi cell of~$\Lambda$ and thus a parallelotope.
For the sake of brevity we use the shorter notations $\cF_q = \cF_{\vc_q}$, $\cC_q = \cC_{\vc_q}$, $c(q)=c(\vc_q)$, and $\chi(q)=\chi(\vc_q)$.
The exact correspondences between the various notions in the languages of lattices and quadratic forms are as follows.

\begin{lem}
\label{lem:compactness-parallelotopes}
Let $q:\R^n\to\R$ be a positive definite quadratic form defined by $q(x)=x^\intercal A^\intercal A x$, for some invertible matrix $A \in \R^{n \times n}$.
Moreover, write $\Lambda=A\Z^n$ for the lattice generated by~$A$.
Then,
\begin{enumerate}[\indent $i)$]

 \item $\vc_\Lambda = A \, \vc_q$ and $\cF_\Lambda = A \, \cF_q$,

 \item $c(q)=c(\Lambda)$ and $\chi(q)=\chi(\Lambda)$.

\end{enumerate}
\end{lem}

\begin{proof}
For \romannumeral1), observe that
\begin{align*}
\vc_q &= \left\{ x \in \R^n : \| Ax \|^2 \leq \| A(x-z) \|^2 \text{ for all }z \in \Z^n \right\}\\
&= \left\{ A^{-1} y \in \R^n : \| y \|^2 \leq \| y - Az \|^2 \text{ for all }z \in \Z^n \right\} = A^{-1} \vc_\Lambda.
\end{align*}
For the second identity, notice that $x \in \cF_q$ if and only if $\vc_q \cap (\vc_q + x)$ is a facet of $\vc_q$, which holds if and only if $A \vc_q \cap (A \vc_q + Ax)$ is a facet of $A\vc_q = \vc_\Lambda$.
Part \romannumeral2) is a direct consequence of these observations.
\end{proof}

The lattice~$D_4 = \{z \in \Z^4 : z_1+\ldots+z_4 \in 2 \Z\}$ plays a crucial role in representing $4$-dimensional lattices whose Voronoi cell is not a zonotope, and thus deserves a detailed study.

\begin{lem}
\label{lem:d4-basis}
Let $y \in \cC_{D_4} \setminus \cF_{D_4}$.
Then there is a basis $B$ of $D_4$ such that
\[
\cC_{D_4} \setminus \{\pm y\} \subseteq \left\{ Bz : \norm{z}_\infty \leq 1 \right\}.
\] 
\end{lem}
\begin{proof}
We start by characterizing the sets $\cF_{D_4}$ and $\cC_{D_4}$.
Since $\pm 2e_i \pm 2e_j \in 2D_4 \subseteq 2\Z^4$ for $1 \leq i < j \leq 4$, it follows that $S = \{z \in \{0, \pm1\}^4 : \norm{z}^2 \in \{2,4\} \}$ contains all vectors $v \neq 0$ that are shortest in their respective co-set $v + 2D_4$.
In fact, due to parity of the coefficients, we have $S = \cC_{D_4}$.
In the proof of Proposition~\ref{prop:compactness}, we saw that $\cF_{D_4} \subseteq \{z \in D_4 : \norm{z}^2 = 2\}$.
For parity reasons of $z \in \cF_{D_4}$, the vectors $z$ and $-z$ are the unique shortest vectors in $z + 2D_4$, hence we actually have $\cF_{D_4} = \{z \in D_4 : \norm{z}^2 = 2\}$.

Now, let $y \in \cC_{D_4} \setminus \cF_{D_4}$ and observe that $D_4^\star = \Z^4 \cup (\frac12 \one + \Z^4)$.
Then
\[
\left\{x \in \R^4 : \,  |e_i^\intercal x| \leq 1, 1 \leq i \leq 4 \right\} \cap \left\{ x \in \R^4 : \,  |y^\intercal x| \leq 2 \right\}
\]
is a facet-description of $Q := \conv \{ \cC_{D_4} \setminus \{\pm y\} \}$.
The inequality $|y^\intercal x| \leq 2$ arises since the vectors $y_i e_i + y_j e_j$, $1 \leq i < j \leq 4$ are contained in $\cC_{D_4} \setminus \{\pm y\}$.
Taking polars, we obtain that
\[
Q^\star = \conv\left(\left\{ \pm e_i : i=1,\dots,4\right\} \cup \{\pm \tfrac12 y\}\right),
\]
and we see that $Q^\star$ contains the dual lattice basis $B^{-\intercal} = \{ \frac12 y, e_1, e_2, e_3 \}$.
Hence, in the spirit of Lemma~\ref{lem:primal_dual}, every vector in $\cC_{D_4} \setminus \{\pm y\}$ is represented with coefficients in $\{\pm 1, 0\}$ in the corresponding primal basis~$B$.
\end{proof}

Observe that Lemma~\ref{lem:d4-basis} is best possible in the sense that $\chi(D_4) = 2$.\footnote{Engel et al.~\cite{engelmichelsenechal2001newgeometric} claim that they computed $\chi(D_4) = 1$, which turns out to be wrong.}
In order to see this, $\conv (\cC_{D_4})^\star$ is the standard crosspolytope, which does not contain a basis of~$D_4^\star$ as any such basis has to contain a vector in $\frac12 \one + \Z^4$. However, after dilating by $2$, we find the basis $\{\frac12 \one,e_2,e_3,e_4\}$ of~$D_4^\star$ (cf.~Lemma~\ref{lem:primal_dual}.)

We now arrive at the desired compactness of four-dimensional lattices.

\begin{cor}
\label{cor:4-dim-compact}
Every lattice of rank at most four has a compact basis.
\end{cor}
\begin{proof}
We have seen in Proposition~\ref{prop:compactness}~$ii)$, that every lattice of rank at most three has a compact basis.
Thus, let $\Lambda = A \Z^4$ be a full-dimensional lattice, and let $q(x)=x^\intercal A^\intercal A x$ be the corresponding quadratic form.
In the case that~$\vc_q$ is a zonotope, we use Lemma~\ref{lem:compactness-parallelotopes} to get that $\vc_\Lambda= A \, \vc_q$ is a zonotope as well, and thus Theorem~\ref{thm:zonotopal} implies that $c(\Lambda)=1$.

If $\vc_q$ is not a zonotope, then Voronoi's reduction theory as applied in Vallentin's thesis~\cite[Ch.~3]{vallentin2003thesis} shows the following:
We can write $\vc_q = \vc_p + Z(U)$, for some positive definite quadratic form $p$ and a set of vectors $U \subseteq \R^4$.
Moreover, $p$ is such that $\vc_p$ is combinatorially equivalent to the $24$-cell.
Up to isometries and scalings, the only lattice whose Voronoi cell is combinatorially equivalent to the $24$-cell is the root lattice $D_4$, defined in Proposition~\ref{prop:compactness}.
This is due to the fact that $D_4$ is what is called a \emph{rigid} lattice.
Therefore, any lattice corresponding to~$p$ agrees with~$D_4$ up to isometries and scalings.

By Lemmas~\ref{lem:compactness-parallelotopes} and~\ref{lem:d4-basis}, this means that for every vector $y \in \cC_{p} \setminus \cF_{p}$, we can find a basis $B$ of~$\Lambda_p:=\Lambda(\vc_p)$ such that every standard vector of~$\vc_p$ apart from $\pm y$ can be represented with coefficients in $\{\pm 1,0\}$.
By the first part of Proposition~\ref{prop:zonotopal-part}, there is a linear map~$\varphi$ such that $\varphi(\Lambda_p) = \Lambda_q :=\Lambda(\vc_q)$, and $\varphi^{-1}(\cF_{q}) =: \cC^\prime \subseteq \cC_p$.
By the second part of Proposition~\ref{prop:zonotopal-part}, we have $\cF_p \subseteq \cC^\prime$.
Since $A\vc_q$ is a Voronoi cell, we have $|\cC^\prime| = |\cF_{q}| \leq 2(2^4 -1) = 30$, whereas $|\cC_p| = |\cC_{D_4}| = 40$ (see the proof of Lemma~\ref{lem:d4-basis}).
Hence we can choose $y \in \cC_{p} \setminus \cC^\prime$, and 
find a basis $B$ of~$\vc_p$ so that all vectors in~$\cC^\prime$ are represented with coefficients in $\{0, \pm 1\}$.
This implies that all vectors in~$\cF_q$ have coefficients in $\{0, \pm 1\}$ when represented in the basis $\varphi(B)$ of~$\vc_q$.
Thus, the lattice~$\Lambda$ has a compact basis as $c(q)=c(\Lambda)$.
\end{proof}

\noindent We summarize the results of this section in Table~\ref{table:compactness}.

\begin{table}[hbt]
\centering
\begin{tabular}{c|l|l}
dimension of $\Lambda$ & \multicolumn{1}{c|}{compactness result}                          & \multicolumn{1}{c}{reference}                                \\ \hline\hline
$n \leq 3$               & $c(\Lambda)=\chi(\Lambda)=1$                & Prop.~\ref{prop:compactness} \& Cor.~\ref{cor:zonotopes-chi-compact}\\ \hline
$n = 4$                    & $c(\Lambda)=1$, but $\chi(D_4)=2$           & Corollary~\ref{cor:4-dim-compact}                   \\ \hline
$n \geq 5$               & $c(\Lambda_n) \geq \lceil\frac{n}{4}\rceil$ & Theorem~\ref{thm:lower_bound}    
\end{tabular}
\caption{Compactness of lattices in small dimensions.}
\label{table:compactness}
\end{table}

\section{Relaxing the basis condition}
\label{sec:relaxed}

The compact representation problem for the set of Voronoi relevant vectors does not need $B$ to be a basis of the lattice $\Lambda$.
In fact, it suffices that we find linearly independent vectors $W=\{w_1,\ldots,w_n\}$ that allow to decompose each Voronoi relevant vector as an integer linear combination with small coefficients.
This is due to the fact that, given a basis, membership to a lattice can be checked in polynomial time.
Thus, in case that the relaxation improves the compactness of the presentation, this additional check is faster than iterating over the larger set corresponding to a~$c(\Lambda)$-compact basis.

\begin{definition}
\label{def:c-reduced_lin_ind}
Let $\Lambda \subseteq \R^n$ be a lattice.
A set of linearly independent vectors $W = \{w_1,\ldots,w_n\} \subseteq \R^n$ is called \emph{$c$-compact for $\Lambda$}, if
\[
\cF_\Lambda \subseteq \left\{ w_1z_1+\ldots+w_nz_n : z \in \Z^n, \, \| z \|_\infty \leq c \right\}.
\]
Moreover, we define
\[
\bar c(\Lambda) = \min\{c \geq 0 : \text{there is a }c\text{-compact set }W\text{ for }\Lambda\}
\]
as the \emph{relaxed compactness constant} of $\Lambda$.
\end{definition}

If every Voronoi relevant vector is an integral combination of $W$, then so is \emph{every} lattice vector.
That is, a $c$-compact set $W$ for $\Lambda$ gives rise to a superlattice $\Gamma = W\Z^n \supseteq \Lambda$.
The relaxed compactness constant and $c(\Lambda)$ are related as follows.

\begin{prop}
\label{prop:c_vs_bar_c}
For every lattice $\Lambda$ in $\R^n$, $n \geq 2$, we have
\[
\bar c(\Lambda) = \lambda_n(Q^\star,\Lambda^\star) \qquad \text{and} \qquad \bar c(\Lambda) \leq c(\Lambda) \leq \frac{n}{2} \, \bar c(\Lambda),
\]
where $Q = \conv( \cF_\Lambda )$ as before.
\end{prop}
\begin{proof}
The identity $\bar c(\Lambda) = \lambda_n(Q^\star,\Lambda^\star)$ follows by arguments analogous to those establishing the equivalence of $i)$ and $ii)$ in Lemma~\ref{lem:primal_dual}.
The inequality $\bar c(\Lambda) \leq c(\Lambda)$ is a direct consequence of the definition of these parameters.

In order to prove that $c(\Lambda) \leq \frac{n}{2} \, \bar c(\Lambda)$, we let $v_1,\ldots,v_n \in (\bar c(\Lambda) \cdot Q^\star) \cap \Lambda^\star$ be linearly independent, and for $1 \leq k \leq n$, we consider the crosspolytope $C_k = \conv\{\pm v_1,\ldots,\pm v_k\}$.
We show by induction that there are vectors $u_1,\ldots,u_n \in \Lambda^\star$ such that (a) $\{u_1,\ldots,u_k\}$ is a basis of the lattice $\Lambda^\star \cap \lin\{v_1,\ldots,v_k\}$, and (b) $u_k \in \max\{\frac{k}{2},1\} \cdot C_k$, for every $1 \leq k \leq n$.
This then implies that $\{u_1,\ldots,u_n\}$ is a basis of $\Lambda^\star$ contained in $\frac{n}{2} \, C_n \subseteq \frac{n}{2} \, \bar c(\Lambda) \, Q^\star$.
Hence, $c(\Lambda) \leq \frac{n}{2} \, \bar c(\Lambda)$, as desired.

First, at least one of the vectors $v_1,\ldots,v_n$ must be primitive, say $v_1$.
Then, setting $u_1 = v_1$ gets the induction started.
Now, let us assume that we found $u_1, \dots, u_{k-1}$ satisfying (a) and~(b).
Let $y \in (\Lambda^\star \cap \lin\{v_1,\ldots,v_k\})^\star$ be a primitive vector orthogonal to $\lin\{u_1,\ldots,u_{k-1}\}$ and such that $y^\intercal v_k \neq 0$.
If $|y^\intercal v_k| = 1$, then $u_k = v_k \in C_k$ complements $\{u_1,\ldots,u_{k-1}\}$ to a basis of $\Lambda^\star \cap \lin\{v_1,\ldots,v_k\}$.
So, we may assume that $|y^\intercal v_k| \geq 2$.
Every translate of $\frac{k-1}{2} \, C_{k-1}$ within $\lin\{v_1,\ldots,v_{k-1}\}$ contains a point of $\Lambda^\star$.
In particular, there is a vector $u_k \in \Lambda^\star$ contained in $\frac{1}{y^\intercal v_k} \, v_k + \frac{k-1}{2} \, C_{k-1}$.
By construction, $u_k$ complements $\{u_1,\ldots,u_{k-1}\}$ to a basis of $\Lambda^\star \cap \lin\{v_1,\ldots,v_k\}$, and since $\vert \frac{1}{y^\intercal v_k} \vert \leq \frac12$, we get that $u_k \in \frac12 \, C_k + \frac{k-1}{2} \, C_{k-1} \subseteq \frac{k}{2} \, C_k$.
\end{proof}

The relaxation to representing $\cF_\Lambda$ by generating sets rather than by lattice bases may reduce the respective compactness constant drastically.
In fact, the quadratic upper bound in Theorem~\ref{thm:cn-upper-bound} improves to $\mathcal{O}(n \log{n})$.
However, there is still a class of lattices that shows that in the worst case the relaxed compactness constant can be linear in the dimension as well.
In combination with Theorem~\ref{thm:lower_bound}, the second part of the following result moreover shows that the factor $n/2$ in Proposition~\ref{prop:c_vs_bar_c} is tight up to a constant.

\begin{prop}
\label{prop:lambda<2}\
\begin{enumerate}[\indent $i)$]
 \item For every lattice $\Lambda \subseteq \R^n$, we have $\bar c(\Lambda) \in \mathcal{O}(n \log{n})$.

 \item For $a=\lceil \frac{n}{2} \rceil$, let $\Lambda_n = \Lambda_n(a)$ be the lattice defined in~\eqref{eqn:Lambda_n}.
 For every $n \in \N$, we have $\bar c(\Lambda_n) \leq 3$, whereas $c(\Lambda_n) \geq \lceil \frac{n}{4} \rceil$, for $n \geq 4$.

 \item There are self-dual lattices $\Lambda \subseteq \R^n$ with relaxed compactness constant $\bar c(\Lambda) \in \Omega (n)$.
\end{enumerate}
\end{prop}
\begin{proof}
$i)$: The polytope $Q = \conv(\cF_\Lambda)$ is centrally symmetric, all its vertices are points of $\Lambda$, and $\inter(Q) \cap \Lambda = \{0\}$.
Therefore, we have $\lambda_1(Q,\Lambda) = 1$.
Proposition~\ref{prop:c_vs_bar_c} and the transference theorem of Banaszczyk~\cite{banaszczyk1996inequalities} thus imply that there is an absolute constant $\gamma > 0$ such that
\begin{align}
\bar c(\Lambda) = \lambda_n(Q^\star,\Lambda^\star) = \lambda_1(Q,\Lambda) \cdot \lambda_n(Q^\star,\Lambda^\star) \leq \gamma\,n\log{n}.\label{eqn:flatness_thm}
\end{align}

$ii)$: In view of Proposition~\ref{prop:c_vs_bar_c}, we have to find $n$ linearly independent points of $\Lambda_n^\star$ in $3 \, Q^\star$.
To this end, we define $y_i := \frac{1}{a}(e_i-e_n)$, for $1 \leq i \leq n-1$.
Furthermore, let $y_n=\frac{1}{a}\one$, if $n$ is even, and $y_n=\left(\{\frac{1}{a}\}^{n-1},\frac{2}{a}\right)$, if $n$ is odd.
We claim that the vectors $y_1,\ldots,y_n$ do the job.

First of all, they are clearly linearly independent, and the description~\eqref{eqn:dualLambda_n} shows that all these vectors belong to $\Lambda_n^\star$.
Now, recall that $Q^\star = \{y \in \R^n : y^\intercal v \leq 1 \text{ for all } v \in \cF_{\Lambda_n}\}$.
By Lemma~\ref{lem:vor-rel-vectors}, a Voronoi relevant vector $v$ of $\Lambda_n$ either equals $\pm \one$ or is contained in $v \in \{a-\ell,-\ell\}^n$, for some suitable $\ell \in \N$.
Consider first the vectors $y_i$, for $1 \leq i \leq n-1$.
We have $\one^\intercal y_i = 0$, and for any $v \in \{a-\ell,-\ell\}^n$ holds $v^\intercal y_i = \frac{1}{a}(v_i - v_n)$ which equals $0$, if $v_i = v_n$, and it equals $\pm 1$, if $v_i \neq v_n$.
Thus, in fact $y_1,\ldots,y_{n-1} \in Q^\star$.

Regarding the remaining vector $y_n$, we observe that $\one^\intercal y_n = 2$, independently of the parity of the dimension~$n$.
Thus, let $v \in \{a-\ell,-\ell\}^n$, and note that $\ell \in \{\lfloor \frac{ak}{n} \rfloor, \lceil \frac{ak}{n} \rceil\}$, where $k=|\{i : v_i=a-\ell\}|$.
Since $-\ell \leq a$ and $a-\ell \leq a$, we have
\begin{align*}
y_n^\intercal v &\leq \tfrac{1}{a} \left( k(a-\ell) - (n-k)\ell + a \right) = \tfrac{1}{a} \left( ka - n \ell + a \right) \\
&\leq \tfrac{1}{a} \left( ka - n (\tfrac{ak}{n} - 1) + a \right) = \frac{n+a}{a} \leq 3,
\end{align*}
and similarly $y_n^\intercal v \geq -3$.
Hence, $y_n \in 3 \, Q^\star$, finishing the proof.

$iii)$: Let $\Lambda$ be a self-dual lattice and let $\vc_\Lambda$ be its Voronoi cell.
Each Voronoi relevant vector~$v \in \cF_\Lambda$ provides a facet of $\vc_\Lambda$ via the inequality $v^\intercal x \leq \frac12\norm{v}^2$, as well as a facet of $Q^\star$ via the inequality $v^\intercal x \leq 1$ (this indeed defines a facet, as a vertex $v$ of $Q$ always induces a corresponding facet of the polar $Q^\star$).
As $\norm{v} \geq \lambda_1(B_n,\Lambda)$, for every $c < \lambda_1(B_n,\Lambda)^2$, we have that $c \cdot Q^\star$ is contained in the interior of twice the Voronoi cell of $\Lambda^\star = \Lambda$, and hence contains no non-trivial dual lattice point.
Therefore, $\bar c(\Lambda) \geq \lambda_1(B_n,\Lambda)^2$.

Conway \& Thompson (see~\cite[Ch.~2, \S 9]{milnorhusemoller1973symmetric}) proved that there are self-dual lattices $\Lambda$ in~$\R^n$ with minimal norm
\[
\lambda_1(B_n,\Lambda) \geq \left\lfloor \frac{1}{\sqrt{\pi}}\left(\frac{5}{3}\Gamma\left(\frac{n}{2}+1\right)\right)^{\frac{1}{n}}\right\rfloor.
\]
Stirling's approximation then gives that $\bar c(\Lambda) \in \Omega(n)$.
\end{proof}

Based on the common belief that the best possible upper bound in~\eqref{eqn:flatness_thm} is linear in~$n$, we conjecture the following:

\begin{conj}
The compactness constants are linearly bounded, that is
\[
\bar c(\Lambda) \in \mathcal{O}(n) \qquad \text{and also} \qquad c(\Lambda) \in \mathcal{O}(n),
\]
for every lattice $\Lambda \subseteq \R^n$.
\end{conj}

\section{Algorithmic point of view}
\label{sec:algorithmic}

When it comes to computing a $c(\Lambda)$-compact basis for $\Lambda$, not much is known.
Lemma~\ref{lem:primal_dual} suggests to take the polar of $\conv (\cF_\Lambda)$, and then to look for a dual basis in a suitable dilate thereof.
However, in order to do this, we need a description of the Voronoi relevant vectors in the first place.
Even if we are only interested in an $(n \cdot c(\Lambda))$-compact basis, it is not clear how to benefit from the allowed slack.

In the following, we rather discuss how to incorporate an already known $c$-compact basis into the algorithm of Micciancio \& Voulgaris~\cite{micciancio2013deterministic}.

\subsection*{The Micciancio-Voulgaris algorithm}

The algorithm consists of two main parts.
In a preprocessing step, it computes the Voronoi cell $\vc_\Lambda$, which can be done in time $2^{\mathcal{O}(n)}$ in a recursive manner.
As a $c$-compact basis already grants a superset of $\cF_\Lambda$, we do not recall the details of this first part.

Once the Voronoi cell $\vc_\Lambda$ is computed, a vector $p \in \Lambda$ is closest to $t$ if and only if $t-p \in \vc_\Lambda$.
Bearing this in mind, the idea is to iteratively subtract lattice vectors from $t$ until the condition holds.

But why do we only need $2^{\mathcal{O}(n)}$ iterations?
Let us assume for now that~$t$ is already rather close to $0$, say $t \in 2 \, \vc_\Lambda$.
Let $p$ be a Voronoi relevant vector whose induced facet-defining inequality is violated by $t$, this means $p^\intercal t > \frac12 \norm{p}^2$.
Micciancio \& Voulgaris show that $t-p$ is still contained in $2 \, \vc_\Lambda$, and is strictly shorter than~$t$.
Hence, for going from $t \in 2 \, \vc_\Lambda$ to some $t^\prime = t - w \in \vc_\Lambda$, for $w \in \Lambda$, the number of iterations we need is bounded by the number of level sets of the norm function that have a point in $2 \, \vc_\Lambda \cap (t + \Lambda)$.
This number turns out to be at most~$2^n$.

If $t$ is further away, that is $t \notin 2 \, \vc_\Lambda$, let $k$ be the smallest integer such that $t \in 2^k \vc_\Lambda$.
Then, we can apply the above method to the lattice $\Lambda' = 2^{k-1} \Lambda$, and find $w \in \Lambda' \subseteq \Lambda$ such that $t-w \in \vc_{\Lambda'} = 2^{k-1} \vc_\Lambda$.
Doing this recursively yields that after $2^n k$ iterations, we moved $t$ into $\vc_\Lambda$.
Note that $k$ is polynomial in the input size.
More sophisticated arguments allow to limit $k$ in terms of $n$ only, or to decrease the number of iterations to weakly polynomial, as presented in~\cite{bonifasdadush2015short}.

\begin{cor}
\label{cor:efficientWcompact}
Assume we are given a $c$-compact basis $B$ of a lattice $\Lambda \subseteq \R^n$.
For any target point $t \in \R^n$, a closest lattice vector to $t$ can be found in time $\mathcal{O} ((2c+1)^n \, 2^n \poly(n))$ and space polynomial in the input size.
\end{cor}
\begin{proof}
Theorem 4.2 and Remark 4.4 in~\cite{micciancio2013deterministic} state that a closest vector can be found in time $\mathcal{O} (\vert V \vert \cdot 2^n \poly(n))$, where $V$ is a superset of the Voronoi relevant vectors~$\cF_\Lambda$.
We set $V = \left\{ Bz : z \in \Z^n, \, \| z \|_\infty \leq c \right\} \supseteq \cF_\Lambda$.

The reduction to polynomial space follows from~\cite[Rem.~4.3]{micciancio2013deterministic}:
Their algorithm may need exponential space because they store $\cF_\Lambda$.
As a subset of $V$ it is however described just by the polynomial-size data $(B,c)$.
\end{proof}

The Micciancio-Voulgaris algorithm naturally can be presented as an algorithm for the 
Closest Vector Problem with Preprocessing (CVPP).
In this variant of CVP, we may precompute the lattice for an arbitrary amount of time and store some additional information.
Only then the target vector is revealed to us, and we are allowed to use the information we gathered before to speed up the process of finding a closest vector.
This is motivated by the fact that in practice, we might have to compute the closest vector for several target vectors, but always on the same lattice.
Hence, we happily spend more time for preprocessing, when we are able to vastly benefit from the additional information.

Considered in this setting, our results compress the information after the preprocessing step into polynomial space.
However, it is unclear how to compute a $c(\Lambda)$-compact basis \emph{without} computing the Voronoi cell first.

\begin{prob}
Can we compute a basis $B$ of $\Lambda$ that attains $c(\Lambda)$ in single-exponential time and polynomial space?
\end{prob}

The fact that every zonotopal lattice has a compact basis is especially interesting.
McCormick, Peis, Scheidweiler \& Vallentin can solve the Closest Vector Problem in polynomial time on a zonotopal lattice, provided it is given in a certain format.\footnote{At the time of writing there is no preprint available (personal communication with Frank Vallentin).
}
Another related result is due to McKilliam, Grant \& Clarkson~\cite{mckilliam2014finding}, who provide a polynomial time algorithm for lattices of Voronoi's first kind, provided an obtuse superbasis is known.
One could wonder whether our representation also allows for solving CVPP faster (measuring only the time after the preprocessing).
However, McKilliam et al.\ use additional combinatorial properties of an obtuse superbasis that are in general not even fulfilled for a $1$-compact basis.
In fact, Micciancio~\cite{micciancio2001hardness} showed that if CVPP can be solved in polynomial time for arbitrary lattices, then $\mathrm{NP} \subseteq \mathrm{P/poly}$ and the polynomial hierarchy collapses.


\subsection*{Acknowledgments}
We thank Daniel Dadush and Frank Vallentin for helpful remarks and suggestions.
In particular, Daniel Dadush pointed us to the arguments in Theorem~\ref{thm:cn-upper-bound} that improved our earlier estimate of order $\mathcal{O}(n^2 \log{n})$.
We are grateful to the referees for their valuable suggestions, questions, and comments.

\bibliographystyle{amsplain}
\bibliography{bibliography}

\end{document}